\documentclass[onefignum,onetabnum]{siamart171218}[review]

\usepackage{amsfonts}
\usepackage{bbold}
\usepackage{graphicx}
\usepackage{epstopdf}
\usepackage{algorithmic}
\usepackage{amssymb,amsmath}
\usepackage{changepage}
\usepackage{color,graphicx,float}
\usepackage{bm,enumerate}
\usepackage{alltt,listings,tikz}
\usepackage[all]{xy}
\usepackage{cite}
\usepackage{hyperref}

\newcommand{\add}[1]{{\color{black} #1}}

\newcommand{\floor}[1]{\lfloor #1 \rfloor}

\newcommand{\bx}{\bm{x}}
\newcommand{\by}{\bm{y}}
\newcommand{\bxe}{\bm{x}_e}
\newcommand{\bye}{\bm{y}_e}
\newcommand{\bxte}{\bm{x}_{\tilde{e}}}
\newcommand{\byte}{\bm{y}_{\tilde{e}}}
\newcommand{\bR}{\mathbb{R}}

\ifpdf
  \DeclareGraphicsExtensions{.eps,.pdf,.png,.jpg}
\else
  \DeclareGraphicsExtensions{.eps}
\fi

\newsiamremark{remark}{Remark}
\newsiamremark{hypothesis}{Hypothesis}
\newsiamremark{conjecture}{Conjecture}
\crefname{hypothesis}{Hypothesis}{Hypotheses}
\newsiamthm{claim}{Claim}
\usepackage{amsopn}

\title{A density description of a bounded-confidence model of opinion dynamics on hypergraphs}
\author{Weiqi Chu\footnote{Department of Mathematics, University of California, Los Angeles. }
 \and Mason A. Porter\footnotemark[1]\,\,\,\footnote{Santa Fe Institute}
 }

\begin{document}

\maketitle

\begin{abstract}
Social interactions often occur between three or more agents simultaneously. Examining opinion dynamics on hypergraphs allows one to study the effect of such polyadic interactions on the opinions of agents. In this paper, we consider a bounded-confidence model (BCM), in which opinions take continuous values and interacting agents comprise their opinions if they are close enough to each other. We study a density description of a Deffuant--Weisbuch BCM on hypergraphs. We derive a rate equation for the mean-field opinion density as the number of agents becomes infinite, and we prove that this rate equation yields a probability density that converges to noninteracting opinion clusters. Using numerical simulations, we examine bifurcations of the density-based BCM's steady-state opinion clusters and demonstrate that the agent-based BCM converges to the density description of the BCM as the number of agents becomes infinite.
\end{abstract}





\section{Introduction}
\add{People spread and change their opinions through their daily social interactions \cite{jackson2010social}.
Mathematical models of opinion dynamics give quantitative approaches to study how the opinions of people and other agents evolve as dynamical processes on networks \cite{noor2020,porter2016}. 
Such models have given insight into a variety of topics, including decision-making \cite{urena2019review}, opinion formation \cite{fiorina2008political,jalili2013social}, and rumor spreading \cite{friedkin2017truth}.}

Models of opinion dynamics can have discrete-valued opinions or continuous-valued opinions \cite{noor2020}. Examples of the former include voter models \cite{redner2019}; examples of the latter include the DeGroot model \cite{degroot1974reaching}, the Friedkin--Johnsen model \cite{friedkin1990social}, and bounded-confidence models (BCMs) \cite{lorenz2007continuous}. In a BCM, the agents that interact with each other compromise their opinions by some amount if and only if the difference between their opinions is less than some threshold, which is known as the confidence bound. 
The opinions of the agents in a BCM can take continuous real values from a finite interval, the entire real line, or a higher-dimensional space.
The compromise mechanism in a BCM is motivated by the idea of ``selective exposure" from psychology; people tend to favor views that are close to their beliefs and to avoid cognitive dissonance \cite{frey1986recent,sears1967selective}.
Personalized recommendations on online platforms also reinforce selective exposure by suggesting content (e.g., YouTube videos) that is based on prior consumed content \cite{knobloch2005impact,kakiuchi2018influence,hossein2020}.

In the last two decades, there have been many studies of BCMs, which have built on pioneering research on the Deffuant--Weisbuch (DW) \cite{deffuant2000mixing,weisbuch2002meet} and Hegselmann--Krause (HK) \cite{hegselmann2002opinion} models. In both the DW and HK models, agents adjust their opinions at discrete time steps when their opinions are sufficiently close to their neighbors. The DW model has asynchronous opinion updates in which one pair of adjacent nodes interacts; these nodes update their opinions if they are sufficiently close to each other.
By contrast, in the HK model, all potential opinion updates of the nodes occur simultaneously (i.e., opinion updates are synchronous). 
BCMs have been generalized in a variety of ways, such as by incorporating heterogeneous confidence bounds or heterogeneous compromise tendencies \cite{pluchino2006compromise,weisbuch2002meet}, randomness in opinion updates in the form of endogenous opinion evolution \cite{baccelli2017pairwise}, and special nodes (such as media nodes) whose update rules are different from those of other nodes \cite{brooks2020model}.

All of the above BCMs encode interactions between agents in the form of dyadic (i.e., pairwise) relationships. 
However, many social interactions are polyadic (i.e., they involve three or more agents) \cite{durlauf2010social}. 
For example, people can discuss their opinions through group texts, live conversations in a video conference call, and in small in-person meetings. 
\add{Polyadic interactions occur both in humans \cite{alvarez2021evolutionary} and in other animals \cite{shemesh2013high,letten2019mechanistic}.
In a recent study~\cite{lambiotte2019networks}, Lambiotte et al. illustrated that pairwise interactions cannot explain the complex non-Markovian dynamics (such as directional passenger flows) in the London Underground transportation system (i.e., ``The Tube'').}
One way to incorporate polyadic interactions is by studying dynamical processes on hypergraphs \cite{battiston2020networks,bick2021higher,de2020social}.
The edges of a graph connect only two nodes (or connect a single node to itself, in the case of a self-edge), whereas the hyperedges of a hypergraph can connect any number of nodes to each other \cite{ouvrard2020hypergraphs,chodrow2020configuration}; they thereby allow one to study collective interactions between arbitrarily many agents.
Two recent papers generalized BCMs to hypergraphs \cite{hickok2022bounded,schawe2022higher}. In one of them~\cite{hickok2022bounded}, Hickok et al. showed that polyadic interactions in a BCM can enhance the convergence to opinion consensus and that BCMs on hypergraphs can possess qualitative dynamics, such as opinion jumps, that do not occur for BCMs on ordinary graphs.

In the study of BCMs, one compelling question is whether the opinions of the nodes of a network eventually reach a consensus state (with one major cluster of opinions), a polarized state (with two major clusters), or a fragmented state (with three or more major clusters) \cite{li2013consensus,meng2018opinion}. It is also important to consider how long it takes to reach a steady state. Given a random initial configuration of opinions, one can examine a BCM as a multi-particle system and perform Monte Carlo simulations to determine a steady-state opinion distribution \cite{pineda2009noisy,fennell2021generalized}. However, direct simulations are computationally expensive when the number of agents is large. For example, numerical observations have suggested that the convergence time of the DW model on cycle graphs grows approximately exponentially with respect to the number of nodes \cite{meng2018opinion}. Additionally, when considering random graphs or networks with random initial opinions, one needs to perform simulations with many realizations to mitigate sampling errors and to estimate the expectations of quantities by computing sample means \cite{fu2014opinion,huang2018effects}.

To study BCM dynamics in a system that involves randomness, an alternative approach to direct simulations is to model the probability density of opinion states using an integro-differential equation \cite{fortunato2005vector,carro2013role,ben2003bifurcations}. 
Ben-Naim et al.~\cite{ben2003bifurcations} modeled the probability density of opinions with the equation
\begin{equation}     \label{eq: node2density}
    \frac{\partial}{\partial t} P(x,t) ={\int\!\!\!\int}_{\!\!\!_{|x_1 - x_2| < 1}} \mkern-50mu P(x_1,t)P(x_2,t) \left[ \delta\left(x-\frac{x_1 + x_2}{2}\right)-\delta(x - x_1) \right] \mathrm{d}x_1 \, \mathrm{d}x_2\,.
\end{equation}
\add{In this density-based BCM, the opinions of the nodes lie in a one-dimensional (1D) space, and $P(x, t)\,\mathrm{d}x$ denotes the fraction of agents with opinions in the interval $(x, x + \mathrm{d}x)$ at time $t$. 
If the opinion difference between two agents is less than $1$, they compromise their opinions to precisely the middle of their current opinions. This middle-point interaction rule leads to a gain term of the opinion density at the compromise opinion $(x_1 + x_2)/2$ and loss terms at the original opinions $x_1$ and $x_2$.}

In the present paper, we study a density-based BCM on hypergraphs.
We consider a DW model on a hypergraph in a discrete-time setting and obtain a continuous-time rate equation in a mean-field limit as the network size (i.e., the number of agents in the network) becomes infinite. 
The rate equation \eqref{eq: node2density} of Ben-Naim et al.~\cite{ben2003bifurcations} is a special case of our model. We prove that the solution of our rate equation is a probability density and that it converges to noninteracting, isolated opinion clusters at an exponential rate. 
These results are consistent with both steady-state behavior and convergence properties in agent-based models~\cite{lorenz2005stabilization}.
As a case study, we use a special type of hypergraph in which each hyperedge has three nodes --- such hypergraphs occur, for example, in the study of folksonomies \cite{ghoshal2009random} --- and numerically examine the steady-state distributions of opinion clusters for different discordance functions, which are analogous to the confidence bounds of dyadic BCMs. 
We observe numerically for both bounded and unbounded opinion distributions that opinion clusters undergo a periodic sequence of bifurcations as we increase the variance of the initial opinion distribution. 
We also illustrate through Monte Carlo simulations that the agent-based DW model converges to our density-based DW model as we increase the number of agents.

Our paper proceeds as follows. In Section \ref{sec: density_evolution}, we review the agent-based DW model on hypergraphs and derive the rate equation of a single-agent density in the mean-field limit. We also prove several properties of the solution of our density-based DW model and of its steady states. 
In Section \ref{sec: numerics}, we numerically study steady-state opinion distributions with different confidence bounds and different initial distributions, and we compare them with the results of agent-based simulations. We conclude in Section \ref{sec: discussion}. We give proofs of several results in Appendix \ref{app}.


\section{Density evolution of a DW model on a hypergraph} \label{sec: density_evolution}

Consider an unweighted and undirected hypergraph $H = (V,E)$, where $V=\{1,\ldots,N\}$ is the set of nodes and $E$ is the set of hyperedges. Each hyperedge $e$ is a subset of $V$ and represents a relationship that is shared by all nodes in $e$. 

For a given hyperedge $e$, we decompose the opinion state into two parts: $\bx = (\bxe, \bxte)$, where the first part $\bxe = \{x_i\}_{i\in e}$ gives the opinion values of the nodes in the hyperedge $e$ and the second part $\bxte = \{x_i\}_{i\notin e}$ gives the opinion values of the other nodes (i.e., the nodes in the complement set).

Analogous to the confidence bounds of dyadic opinion models, we use a discordance function $d_p: \bR^{|e|}\rightarrow \bR_{\ge 0}$ to measure the heterogeneity of the opinions that are associated with a hyperedge $e$. We define the discordance function as an averaged deviation from a group mean in terms of the $L_p$ norm:
\begin{equation}     \label{eq: discordancep}
    d_p(\bxe) = \alpha_p\left[ \frac{1}{|e|}\sum_{i\in e}|x_i-\overline{\bx}_e|^p\right]^{1/p}\,,
\end{equation}
where $\overline{\bx}_e=\frac{1}{|e|} \sum_{i\in e}x_i$ is the group mean. The factor $1/|e|$ alleviates the disadvantage of hyperedges with more nodes (i.e., ``larger" hyperedges).  
We introduce scaling constants $\alpha_p$ to reduce the discrepancy between different choices of the parameter $p$. In Section \ref{sec: numerics}, we discuss the selection of $\alpha_p$. 
Other choices of discordance functions include the $L_\infty$ norm \cite{schawe2022higher} and the sample variance \cite{hickok2022bounded}.


\subsection{Rate equation for the discrete-time dynamics and its continuous-time extension}
In the DW model on a hypergraph, at each time step, we choose one hyperedge $e$ randomly with probability $p_e$ and update the opinions of its nodes to the mean opinion in $e$ if and only if the discordance is less than a threshold. 
That is,
\begin{equation}
    x_i(n+1) = 
    \begin{cases}
      \overline{\bx}_e(n) & \text{if $i\in e$ \,and\, $d_p(\bxe(n))<c$} \\
      x_i(n) & \text{otherwise}\,,
    \end{cases}
\label{eq: agent-iterations}
\end{equation} 
where $c > 0$ is the confidence bound. The sequence $\{\bx(n)\}$ forms a discrete-time Markov process. 
We state the precise dynamics of the probability density of $\{\bx(n)\}$ in Theorem \ref{thm: Pdiscrete}, which we prove in Appendix \ref{sec: thmPdiscrete}.

\begin{theorem}
\label{thm: Pdiscrete}
The DW model \eqref{eq: agent-iterations} induces a discrete-time Markov chain $\bx(n)$, with $n=0,1,\ldots$, in the continuous state space $\bR$.
Let $P(\bx,n)$ be the probability density of $\bx(n)$. The density $P(\bx,n)$ satisfies the relation
\begin{equation}
    P(\bx,n+1)-P(\bx,n) = \sum_{e\in E} p_e\int_{d_p(\by_e)<c} \mkern-40mu P(\bye,\bxte,n)\left[\delta(\bxe-\overline{\by}_e)-\delta(\bxe-\bye)\right] \mathrm{d}\bye\,.
    \label{eq: discrete-time}
\end{equation}
\end{theorem}

\add{Theorem \ref{thm: Pdiscrete} gives an exact density description of the agent-based DW model \eqref{eq: agent-iterations} on a hypergraph and provides an approach to simulate agent-based dynamics in terms of associated evolutions of probabilities.}
Solving \eqref{eq: discrete-time} with the initial condition $P(\bx,0)$ yields the probability density function of $x(n)$.
If the agent-based DW model has a random initial state $\bx(0)$, then $P(\bx,0)$ is the probability density of the initial state; if the model has a deterministic initial state $\bx(0)$, then $P(\bx,0)$ is an empirical Dirac measure that is associated with the initial state.

Solving equation \eqref{eq: discrete-time} numerically involves integration over a high-dimensional space. One can decrease the dimension of the integration space in \eqref{eq: discrete-time} by one by explicitly integrating the delta functions, but it remains computationally expensive to simulate the dynamics of the joint distribution of $\bx(t)$, which lives in a space whose dimension is the network size $N$ (i.e., the number of nodes). 
Therefore, we reduce the computational expense by employing a mean-field approximation of the rate equation (see Section \ref{rate}). The reduced rate equation only requires the solution of a single-agent density function, which is one-dimensional.

It is also convenient to treat time as a continuous variable. Some models assume that opinions change continuously with time \cite{sahasrabuddhe2021modelling} and do not allow sudden opinion changes at any instant of time. However, we consider a probability density $f(\bx,t)$ as a continuous extension that is induced by the discrete-time density $P(\bx,n)$. To do this, we write
\begin{equation}\label{defin}
    f(\bx,t) = P(\bx,\floor{\gamma t})\,, \quad t \ge 0\,,
\end{equation}
where $\floor{\cdot}$ is the floor function and the constant $\gamma$ indicates the speed of opinion updates. 
The induced model, which describes the continuous-time probability density $f(\bx,t)$, has $\gamma$ updates of opinion states in the time interval $[t,t+1)$.

With the definition \eqref{defin}, the probability density $f(\bx,t)$ is not continuous with respect to $t$. Specifically, with $\tau = 1/\gamma$, a direct computation yields
\begin{equation}     
\label{eq: fequation}
    \frac{f(\bx,t+{\tau})-f(\bx,t)}{{\tau}} = \gamma \sum_{e\in E} p_e\int_{d_p(\by_e)<c} \mkern-40mu f(\bye,\bxte,t)\left[\delta(\bxe-\overline{\by}_e)-\delta(\bxe-\bye)\right] \mathrm{d}\bye\,.
\end{equation}
If we take the limit $\tau\rightarrow 0$ in \eqref{eq: fequation}, the left-hand side yields the time derivative of the opinion density $f(\bx,t)$, but the right-hand side becomes infinite because $\gamma \rightarrow \infty$. 
When $\tau \rightarrow 0$ (i.e., $\gamma \rightarrow \infty$), agents interact with each other at an infinitely fast rate. If a network has a finite number of agents, it instantaneously reaches a steady state.
However, the right-hand side of \eqref{eq: fequation} becomes well-defined in the limit $\tau \rightarrow 0$ (with an appropriate scaling parameter $\gamma$) when the number of agents in a network simultaneously tends to infinity. 
We thus consider a mean-field limit in which the network size becomes infinite.


\subsection{Rate equation in a mean-field limit}\label{rate}

We introduce a one-point density $g(a,t)$ to examine the opinion distribution of a population of equivalent agents. 
For $a\in\bR$, we define the one-point density
\begin{equation}
    g(a,t) = \frac1N \sum_{i=1}^N \int f(\bx,t) \delta(a-x_i)\,\mathrm{d}\bx\,.
    \label{eq: one-agent-density}
\end{equation}
We derive a governing equation for $g(a,t)$ in Theorem \ref{thm: gequation}.

For notational simplicity, we use finite-difference notation and define
\begin{equation}
    \left[{\tau}\right]h(\cdot,t) := \frac{h(\cdot, t + \tau) - h(\cdot,t)}{{\tau}} \,.
\end{equation}
We use $f(\cdot,t)$ with an appropriate argument to denote the marginal densities of $f(\bx,t)$. For example, $f(\bxe,t)$ represents the marginal density after integrating out the variables $\bxte$ that are associated with the complement set. That is, $f(\bxe,t) = \int f(\bx,t)\,\mathrm{d}\bxte$.

\begin{theorem}
\label{thm: gequation}
Assume for all hyperedges $e \in E$ and all times $t \ge 0$ that 
\begin{equation}
    f(\bxe,t) = \Pi_{i\in e}f(x_i,t)\,.
    \label{eq: independence}
\end{equation}
It then follows that the one-point density $g(a,t)$ in equation \eqref{eq: one-agent-density} satisfies
\begin{equation}
    [{\tau}]{g}(a,t) = \frac{\gamma}{N} \sum_{e\in E} p_e \int_{d_p(\by_e)<c} \mkern-30mu \Pi_{i\in e} g(y_i,t) \left[|e|\delta(a-\overline{\by}_e)-\sum_{i\in e}\delta(a-y_i)\right] \mathrm{d}\bye\,,
    \label{eq: one-agent}
\end{equation}
\end{theorem}
where $\gamma$ is the updating rate in \eqref{defin} and $\tau=1/\gamma$.
\medskip
\medskip

We prove Theorem \ref{thm: gequation} in Appendix \ref{sec: thmgequation}. It requires an independence assumption that the opinions of the nodes in a hyperedge $e$ are independent for $t \ge 0$.
In general, this independence assumption does not hold for a network with finitely many agents, but the joint distribution converges to the product of one-agent marginal densities as the number of agents becomes infinite. 
For a DW model on a dyadic graph, G\'{o}mez-Serreno et al.~\cite{gomez2012bounded} proved rigorously (see their Theorem $4.3$) that (1) the law (which is the probability measure that is associated with a random variable or a stochastic process) of opinion processes $\bx(t)$ converges to the law of independent and identically distributed (i.i.d.) processes as the number of agents becomes infinite and that (2) the law of an individual limit process is the unique solution of an integro-differential equation of Kac type. We extend their results to networks with polyadic interactions in Conjecture \ref{thm: conj}.

\begin{conjecture}
\label{thm: conj}
Suppose that $f(\bx,t)$ in \eqref{eq: fequation} is permutation-invariant with respect to $\bx$. (In other words, $f(\bx,t)$ remains the same when we switch any pair of entries of $\bx$.)
If all hyperedges $e \in E$ satisfy $|e| \le M$, then it follows for $t \le T$ that
\begin{equation}
    \left|f(\bxe,t) - \Pi_{i\in e}f(x_i,t)\right| \le \frac{C(M,T)}{N} \,,
    \label{eq: error}
\end{equation}
where $C(M,T)$ is a constant that depends only on $M$ and $T$ and the difference in the absolute value $|\cdot|$ is in the weak sense of using test functions.
\end{conjecture}

\medskip

The inequality \eqref{eq: error} indicates the weak convergence (i.e., the convergence of inner products with test functions) of two probability measures \cite{billingsley2013convergence}. 
Conjecture \ref{thm: conj} states that if the maximum hyperedge size $M := \mathrm{max}\{|e|,~e\in E\}$ is constant (independent on the network size $N$), then the probability density of any $m$ opinions (with $m \le M$) converges to an $m$-fold product of identical one-agent probability densities in the weak sense as the number of agents (i.e., $N$) becomes infinite. 
Specifically, the difference between two densities decays as $O(1/N)$.

We believe that it is possible to rigorously prove Conjecture \ref{thm: conj} by adopting methods from the proofs of Theorem $4.3$ in \cite{gomez2012bounded} and Theorem $3.1$ in \cite{graham1997stochastic}. Both of these papers considered stochastic processes that one can use to describe a dyadic DW model, formulated $f(\bx,t)$ as an empirical measure that is induced by such a stochastic process, and estimated an error bound for the weak convergence in a probability space. In Appendix \ref{sec: thmconj}, we give a strategy to prove Conjecture \ref{thm: conj}.

We take $\gamma = N$ in equation \eqref{eq: one-agent}. With this choice, in one unit of time, the number of selected hyperedges is proportional to the number of agents in a network. We then take the limit $N \rightarrow \infty$ and obtain the mean-field rate equation
\begin{equation}  \label{eq: mean-field}
    \frac{\partial}{\partial t}{g}(a,t) = \sum_{e\in E} p_e \int_{d_p(\by_e)<c} \!\!\!\Pi_{i\in e} g(y_i,t) \left[|e|\delta(a-\overline{\by}_e)-\sum_{i\in e}\delta(a-y_i)\right] \mathrm{d}\bye \,.
\end{equation}
As $N \rightarrow \infty$, the set $E$ has an infinite number of hyperedges and the probability $p_e$ of selecting hyperedge $e$ goes to $0$. However, the sum in \eqref{eq: mean-field} remains finite. 
We can further simplify \eqref{eq: mean-field} if we know the hyperedge set $E$ and the selection probability $p_e$.
\add{For example, if all interactions are dyadic (i.e., $|e| = 2$ for all hyperedges $e$) and $p_e$ is the same for all hyperedges $e$, then equation \eqref{eq: mean-field} is the same as the mean-field model \eqref{eq: node2density} after multiplying the latter by a factor of $2$ on the right-hand side.} 
In Section \ref{sec: numerics}, we simplify \eqref{eq: mean-field} when $E$ is a collection of all unordered triples, and we uniformly randomly choose each hyperedge (i.e., $p_e$ is the same for all $e$).


\subsection{Properties of densities in the mean-field limit}
In this subsection, we study the Cauchy problem for the density-based DW model in equation \eqref{eq: mean-field}, which is a nonlinear integro-differential equation of Kac type. We prove that the solution of this Cauchy problem preserves the basic properties of a probability density when the initial condition of the rate equation is a probability density function (PDF) and converges to noninteracting, isolated clusters as the time $t \rightarrow \infty$.
We use $g(a,t)$ to denote the solution of equation \eqref{eq: mean-field} with the initial condition $g_0(a)$.


\subsubsection{Nonnegativity}
A PDF must be nonnegative. In Theorem \ref{thm: nonnegativity}, we prove nonnegativity when the initial condition $g_0(a)$ is a function in the traditional sense. For the more general case in which $g_0(a)$ is a measure, one can employ a similar proof by using nonnegative test functions.
\begin{theorem}
\label{thm: nonnegativity}
If the initial condition $g_0(a)$ is nonnegative, then $g(a,t)$ is also nonnegative for all times $t > 0$. Furthermore, if $g_0(a)$ is strictly positive at some point $a'$ (i.e., if $g_0(a') > 0$), then $g(a',t) > 0$ for any finite time $t$.
\end{theorem}

\begin{proof}
We proceed by contradiction. Suppose that $t = t^*$ is the earliest time that $g(a,t)$ takes a negative value.
That is,
\begin{equation}
    t^* := \inf_t \left\{\text{there exists some } a^* \text{ such that } g(a^*,t) < 0  \right\}\,.
\end{equation}
Because of the continuity in $t$, we know that $g(a^*,t^*) = 0$, $\frac{\partial}{\partial t} g(a^*,t^*) < 0$, and $g(a,t)\ge 0$ for all $t\le t^*$. 
Therefore, by equation \eqref{eq: mean-field}, it follows for all $t \le t^*$ that
\begin{equation}\label{eq: beta}
	\begin{aligned}
    \frac{\partial}{\partial t}{g}(a,t) &\ge -\sum_{e\in E} p_e \int_{d_p(\by_e)<c} \Pi_{i\in e} g(y_i,t) \left( \sum_{j\in e}\delta(a-y_{j})\right) \, \mathrm{d}\by_e \\
    &= -\sum_{e\in E} |e|p_e \int_{d_p(\by_e)<c} \Pi_{i\in e} g(y_i,t)\delta(a-y_{j}) \, \mathrm{d}y_j \, \mathrm{d}(\by_e\backslash y_{j})\\
    &= - g(a,t)\sum_{e\in E} |e|p_e \int_{d_p((a, \by_e\backslash y_{j}))<c} \Pi_{i\in e, i\neq j} g(y_i,t) \, \mathrm{d}(\by_e\backslash y_{j})\,,
\end{aligned}
\end{equation}
where $j$ is a node in the hyperedge $e$ and $(\by_e\backslash y_{j})$ denotes the opinions of the nodes in $e$ other than $j$. By symmetry, any node $j$ yields the same integral. Because $g(a^*,t^*)=0$, we have that $\frac{\partial}{\partial t}g(a^*,t^*)=0$, which contradicts $\frac{\partial}{\partial t}g(a^*,t^*)<0$.

Define $\beta(a,t) := \sum_{e\in E} |e|p_e \int_{d_p((a, \by_e\backslash y_{j}))<c} \Pi_{i\in e, i\neq j} g(y_i,t)\,\mathrm{d}(\by_e\backslash y_{j})$ in equation \eqref{eq: beta}. 
For any $a$ and $t > 0$, we have
\begin{equation} 
    \frac{\partial}{\partial t}{g}(a,t) \ge -g(a,t)\beta(a,t)\,.
\end{equation}
Because $\beta(a,t)$ is nonnegative, we apply Gr\"{o}nwall's inequality for each fixed $a$ and obtain
\begin{equation}
    g(a,t) \ge g_0(a) e^{-\int_0^t \beta(a,s)\,\mathrm{d}s}\,.
\end{equation}
For any $g_0(a') > 0$, it thus follows that $g(a',t) > 0$ for any finite $t$. 
\end{proof}


\subsubsection{Mass and opinion conservation}

Equation \eqref{eq: mean-field} describes a time-dependent density of the dynamics in \eqref{eq: agent-iterations} in a mean-field limit. Furthermore, we see from \eqref{eq: agent-iterations} that the mean is constant for all times $t$.
In Theorem \ref{thm: conservation}, we prove that $g(a,t)$ has a conserved mass (i.e., $g(a,t)$ integrates to $1$) and conserved mean opinion for all times $t$. We also show in this theorem that $g(a,t)$ has nonincreasing even-order moments. 
We start with a short lemma.

\begin{lemma}
\label{thm: convex}
    Let $\varphi(a)$ be a convex function. The expectation of $\varphi(a)$ with respect to $g(a,t)$ is nonincreasing with time. 
\end{lemma}

\begin{proof}
We take the time derivative of the expectation $\int\! \varphi(a)g(a,t)\,\mathrm{d}a$. A direct computation from \eqref{eq: mean-field} yields
    \begin{equation}
        \frac{\partial}{\partial t} \int\! \varphi(a)g(a,t)\,\mathrm{d}a = \sum_{e\in E}|e|p_e \int_{d_p(\by_e)<c} \!\!\!\Pi_{i\in e} g(y_i,t) \left[\varphi(\overline{\by}_e)-\sum_{i\in e}\frac{\varphi(y_i)}{|e|}\right] \mathrm{d}\bye \,.
        \label{eq: conservation}
    \end{equation}
Because $\varphi(a)$ is convex, it is always the case that $\varphi(\overline{\by}_e)-\sum_{i\in e}\varphi(y_i)/|e| \le 0$, which implies that the integral in \eqref{eq: conservation} is nonincreasing with time. When $\varphi$ is strictly convex, the expectation is strictly decreasing.
\end{proof}

From Lemma \ref{thm: convex}, we know that both centered and uncentered moments of $g(a,t)$ are nonincreasing. We state and prove our results for uncentered moments.

\begin{theorem}
\label{thm: conservation}
Let $\mathcal{M}_k(t)$ be the $k$th uncentered moment of $g(a,t)$. That is,
\begin{equation}
    \mathcal{M}_k(t):=\int g(a,t)a^k\,\mathrm{d}a\,.
\end{equation}
When $k = 0$ (which gives the mass) or $k =1$ (which gives the mean opinion), the moment is constant. Additionally, when $k$ is even and at least $2$, the moment $\mathcal{M}_k(t)$ is nonincreasing with time.
\end{theorem}

\begin{proof}
Let $\varphi(a)=a^k$ in Lemma \ref{thm: convex}. The subtraction in the right-hand side of \eqref{eq: conservation} gives
\begin{equation}
    \left(\frac{\sum_{i\in e}y_i}{|e|}\right)^k - \sum_{i\in e} \frac{y_i^k}{|e|} = 0\,,
\end{equation}
which implies that $\frac{\partial}{\partial t}\mathcal{M}_k(t)=0$ when $k=0$ or $k = 1$.
When $k$ is even and is at least $2$, the function $a^k$ is convex, so the moments are nonincreasing by Lemma \ref{thm: convex}.
\end{proof}

\subsubsection{Convergence to a limit state}
It is important to examine the limit state of the density-based model \eqref{eq: mean-field}. There exist rigorous proofs that some BCMs on graphs \cite{lorenz2005stabilization} and hypergraphs \cite{hickok2022bounded} always converge to some limit state. 
For some BCMs with heterogeneities, such as the model in \cite{weisbuch2003interacting} (which, e.g., includes heterogeneous confidence bounds), researchers have observed convergence numerically but have not proven it rigorously.
In Theorem \ref{thm: convergence}, we present a similar convergence result from the viewpoint of density evolution.

To formulate our convergence statement, we first consider 
the special opinion vector $\by_{|e|}^{\triangle} = (\triangle, 0,\ldots,0)\in \mathbb{R}^{|e|}$. Given a BCM on a network, we define the \emph{minimum isolation distance}
\begin{equation} \label{eq: triangle}
	\triangle_{d_p}^*(c) := \min_{\substack{e \in E {\textrm{ s.t. }}  \\ d_p(\by_{|e|}^{\triangle}) = c}}  \triangle 
\end{equation}
that is associated with its discordance function and confidence bound.

The construction of $\by_{|e|}^{\triangle}$ corresponds to the optimal solution, up to some permutation, of a minimization problem in Lemma \ref{thm: isolationdistance}, which states that the smallest nontrivial discordance is achieved by a vector of the form $\by_{|e|}^{\triangle}$ when the opinions take a set of discrete values. 
Given a discordance function, one can explicitly compute $\triangle_{d_p}^*(c)$ for a BCM on a network. For example, for a BCM on an ordinary graph (i.e., with only dyadic interactions), $\triangle_{d_p}^*(c)$ is equal to the confidence bound $c$. In Section \ref{sec: numerics}, we give the formula for $\triangle_{d_p}^*(c)$ for $3$-uniform hypergrahs (i.e., when $|e| = 3$ for all hyperedges $e$).

\begin{theorem}
\label{thm: convergence}
Let $g_{\infty}(a)$ be the limit density of $g(a,t)$ as time $t \rightarrow \infty$. Such a limit density $g_{\infty}(a)$ always exists, and $g_{\infty}(a)$ is a sum of Dirac delta functions: 
\begin{equation} \label{eq: ginf}
        g_\infty(a) = \sum_{k = 1}^K m_k\delta(a - a_k)\,, 
\end{equation}
where $K$ is the number of opinion clusters, $m_k$ denotes the population mass of the $k$th steady-state opinion cluster, $a_k$ denotes the position (i.e., opinion value) of this opinion cluster, $\sum_{k=1}^K m_k = 1$, and $|a_i - a_j|\ge\triangle_{d_p}^*(c)$ for $i \neq j$. The constant $\triangle_{d_p}^*(c)$ is the minimum isolation distance \eqref{eq: triangle}, and $K$ is finite if the initial density $g_0(a)$ has finite support. 
\end{theorem}

\medskip
\medskip

We prove Theorem \ref{thm: convergence} in Appendix \ref{sec: convergence}. 
This theorem states that the mean-field density converges to a sum of noninteracting Dirac delta functions and that the distance between any two delta functions is at least a constant that is determined completely by the discordance function and the confidence bound.


\section{Numerical results for density-based BCMs on hypergraphs}\label{sec: numerics}

We consider an undirected $3$-uniform hypergraph, so the hyperedge set $E$ is a collection of unordered triples (i.e., $E = \{e: |e| = 3\}$). In each time step, we select a hyperedge uniformly at random.
The mean-field density dynamics satisfies
\begin{equation}
	\begin{aligned}
    		\frac{\partial}{\partial t}{g}(a,t) &= \int_{d_p(\bx)<c} \!\!\! \Pi_{i=1}^3g(x_i,t)\left[ 3\delta\left(a-\frac{1}{3}\sum_{i=1}^3x_i\right)-\sum_{i=1}^3\delta(a-x_i) \right] \mathrm{d}\bx\,, 
    \end{aligned}
    \label{eq: node3density}
\end{equation}
where $d_p$ is the discordance function \eqref{eq: discordancep}.
As indicated by Theorem \ref{thm: convergence}, the limit clusters are isolated from each other by a distance of at least $\triangle_{d_p}^*(c)$, which depends exclusively on the discordance function $d_p$ and the confidence bound $c$.
To diminish the discrepancy between different choices of $d_p$ in our numerical simulations, we introduce constants $\alpha_p$ to ensure that all $d_p$ yield the same minimum isolation distance. Specifically, we use 
\begin{equation}    \label{eq: alpha_p}
    \alpha_p = \left( \frac{3^{p+1}}{2+2^p} \right)^{1/p}
\end{equation}
in equation \eqref{eq: discordancep}. The minimum isolation distance $\triangle_{d_p}^*(c)$ is then equal to the confidence bound $c$.

To reduce the number of dimensions of the rate equation \eqref{eq: node3density}, we integrate the Dirac delta functions and obtain
\begin{equation} \label{eq: node3simplify}
\begin{aligned}
       \frac{\partial}{\partial t}{g}(a,t) &=  9\int_{\widetilde{d}_p((a,x_1,x_2))<c} g(3a-x_1-x_2,t)g(x_1,t)g(x_2,t)\,\mathrm{d}x_1 \, \mathrm{d}x_2 \\
       &\quad - 3g(a,t)\int_{d_p((a,x_1,x_2))<c}
       g(x_1,t)g(x_2,t) \, \mathrm{d}x_1 \, \mathrm{d}x_2\,,
\end{aligned}
\end{equation}
where $\widetilde{d}_p((a,x_1,x_2)) = \alpha_p \left(\frac{|2a-x_1-x_2|^p+|a-x_1|^p+|a-x_2|^p}{3}\right)^{\frac1p}$.
We integrate the rate equation \eqref{eq: node3simplify} using a fourth-order Adams--Bashforth method \cite{zwillinger1998handbook} and shift the density numerically at each time step to conserve the mass. Once the probability distribution has noninteracting sets (i.e., ``patches") of opinions up to some ``separation accuracy", which we take to be $10^{-7}$ in our simulations, we apply numerical integration to the opinion patches independently with adaptive time steps to expedite convergence. 
We integrate the system \eqref{eq: node3simplify} for a sufficiently long time so that the probability densities behave as a sum of Dirac delta functions up to a ``density accuracy", which we take to be $10^{-6}$ in our simulations. In Figure \ref{fig: separator}, we illustrate both our stopping criterion and how to determine the isolated opinion clusters.

\begin{figure}[htbp]
  \centering
  \includegraphics[width=0.8\textwidth]{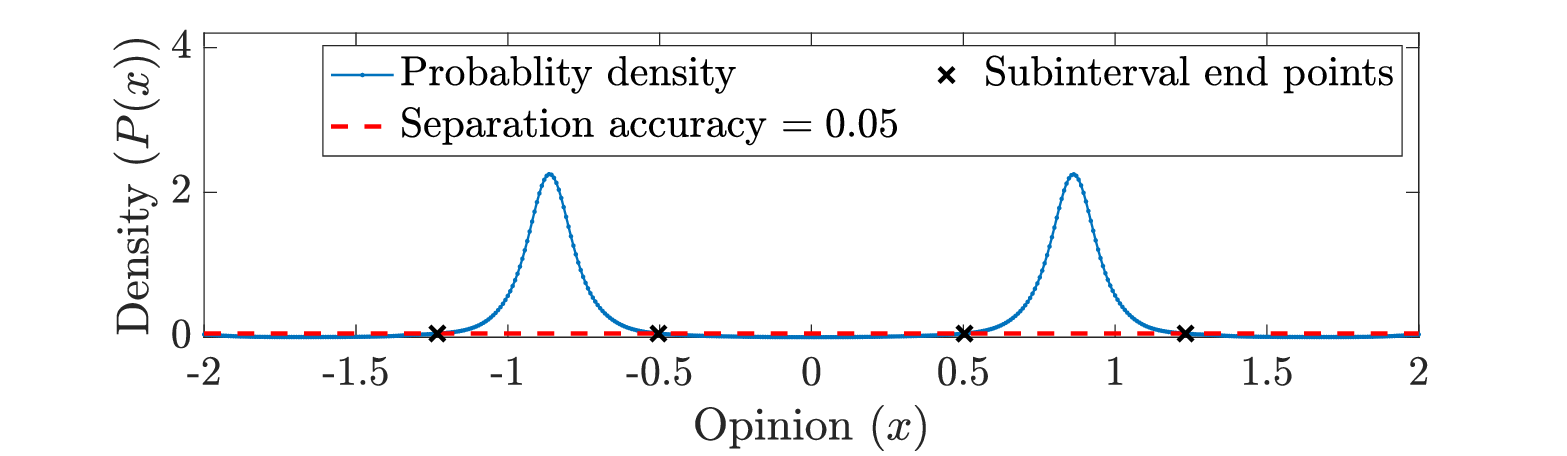}
  \includegraphics[width=0.8\textwidth]{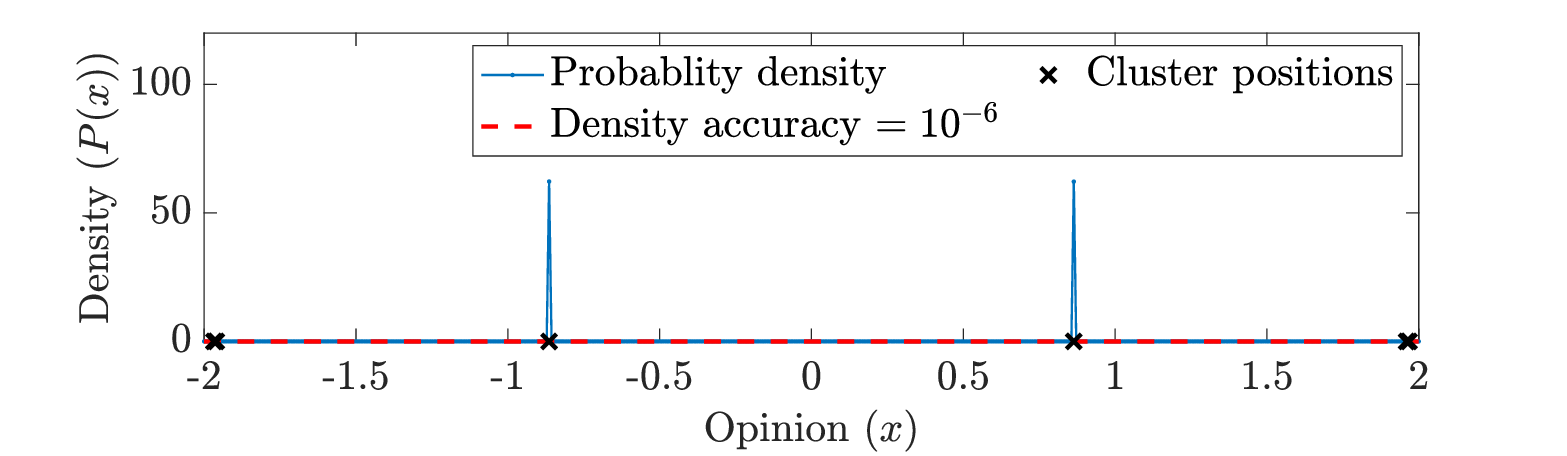}
  \includegraphics[width=0.8\textwidth]{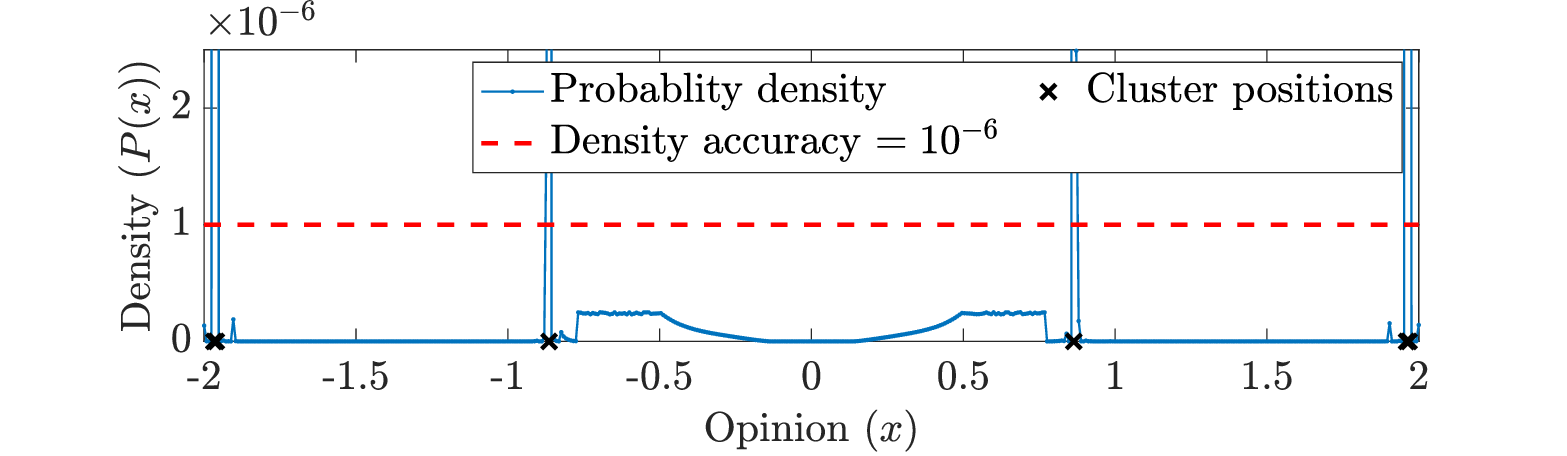}
  \caption{(Top) We divide the depicted probability density into two opinion patches, such that the values of the probability density at $x = -0.5$ and $x = 0.5$ are less than or equal to the separation accuracy $0.05$. 
  The two patches are separated from each other by at least the minimum isolation distance \eqref{eq: triangle}, which is equal to $1$ in the figure. 
  Our computations then reduce to independent simulations on two disjoint intervals. We also show (center) a stopping point of the simulations and (bottom) a magnification of the region near the horizontal axis. We integrate \eqref{eq: node3simplify} until the density profile looks like a sum of Dirac delta functions up to a density accuracy. Specifically, we stop our simulations when any two grid points with probability density values that exceed the density accuracy $10^{-6}$ are in consecutive grid points or are separated by at least the minimum isolation distance. 
  In the center and bottom figures, the mesh size is $0.008$, and the grid points with probabilities that exceed $10^{-6}$ are $x = -1.968$, $x = -1.960$, $x = -0.864$, $x = 0.864$, $x = 1.960$, and $x= 1.968$. 
  These points are either consecutive grid points or they differ by at least $1$, which is the minimum isolation distance.}
  \label{fig: separator}
\end{figure}

In the following subsections, we conduct three sets of numerical experiments. In Section \ref{sec: bifurcation}, we fix the confidence bound to $c = 1$ and observe the steady-state solutions of the mean-field model \eqref{eq: node3density} with different choices of discordance functions $d_p$ for a uniform initial distribution $g_0(a)$.
In Section \ref{sec: gaussian}, we consider two types of initial distributions (uniform and normal) with the same variance, and we numerically obtain the positions of the steady-state opinion clusters.
As we increase the variance, the positions of the clusters undergo a periodic sequence of bifurcations.\footnote{We use the term ``bifurcation'' to refer to a change in qualitative behavior as a parameter crosses some value, rather than demanding the stricter sense of the word ``bifurcation'' from bifurcation theory. 
} The number of steady-state opinion clusters also grows as a periodic sequence.
In Section \ref{sec: meanfieldlimit}, we numerically compare the density-based mean-field model \eqref{eq: mean-field} with agent-based dynamics using Monte Carlo simulations. 


\subsection{Bifurcations patterns with different initial ranges of opinion values} \label{sec: bifurcation}
We fix the confidence bound to $c = 1$ and draw the initial opinions from the uniform distribution on $[-D,D]$. We consider values of $D$ between $0.5$ and $6$ with an increment of $0.1$. We use discordance functions with $p = 0.5$, $p = 1$, and $p = 2$ in \eqref{eq: discordancep}. 
In Figure \ref{fig: varyingD}, {we show the positions (i.e., $a_k$) and the population masses (i.e., $m_k$) of the steady-state opinion clusters} 
for different values of $D$ and $p$. Ben-Naim et al.~\cite{ben2003bifurcations} numerically computed a bifurcation pattern of steady-state cluster positions for a BCM on graphs, and we obtain a similar pattern for our BCM on hypergraphs.

\begin{figure}[htbp]
  \centering
  \includegraphics[width=0.31\textwidth]{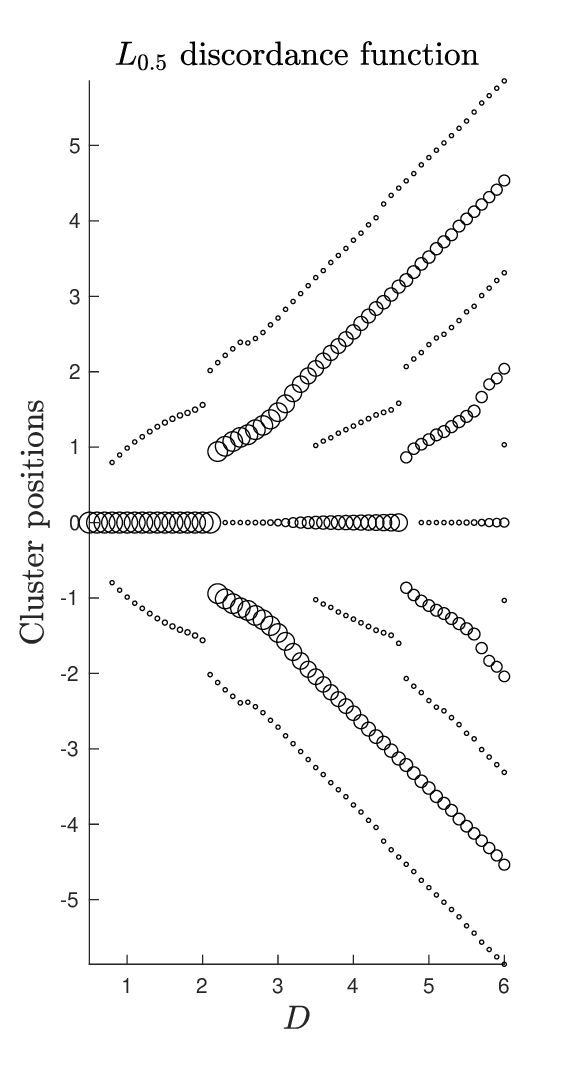}
  \includegraphics[width=0.31\textwidth]{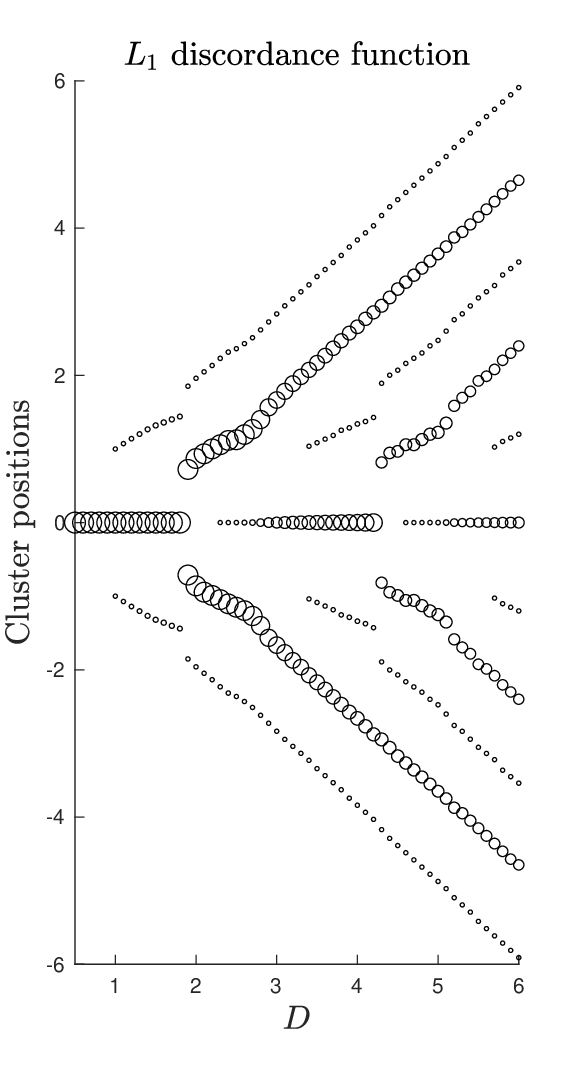}
  \includegraphics[width=0.31\textwidth]{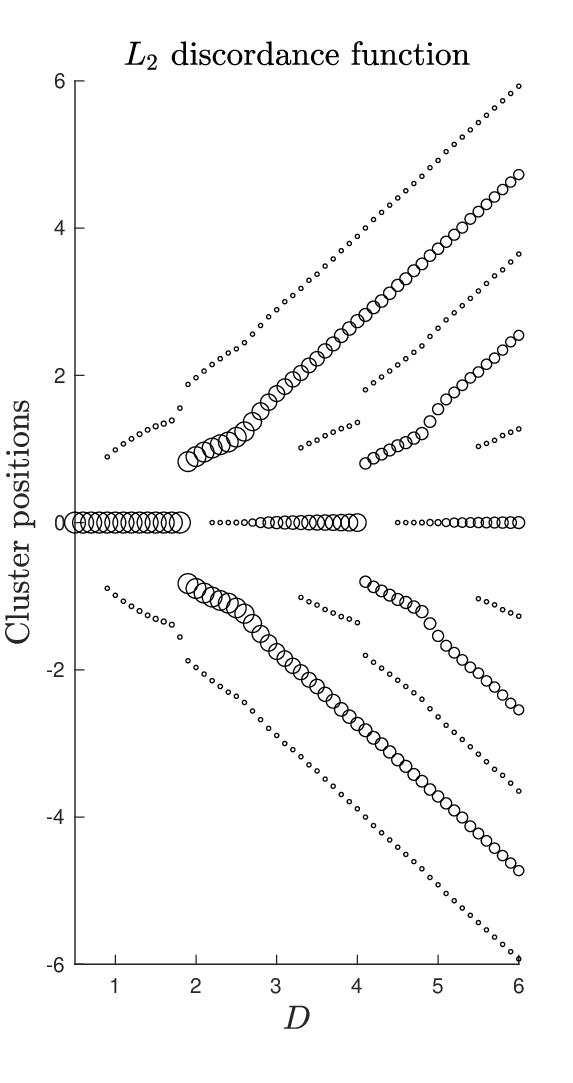} \vspace{-10pt}
  \caption{The positions of steady-state opinion clusters with different initial opinion ranges $[-D,D]$ and discordance function $d_p$ for (left) $p = 0.5$, (center) $p = 1$, and (right) $p = 2$. The sizes of the circles are proportional to the population masses of the steady-state clusters. 
  There are $500$ grid points on the interval $[-D,D]$.}
  \label{fig: varyingD}
\end{figure}

In regions that are not near bifurcation points, we observe in our numerical computations that the positions of the steady-state opinion clusters have a linear relationship with $D$ and that they form parallel straight lines far from the locations at which they nucleate and disappear.
As we proved in Theorem \ref{thm: convergence}, the opinion clusters are separated by a theoretical minimum isolation distance $\triangle_{d_p}^*(c)$, which we set to $1$ (by introducing an appropriate coefficient in \eqref{eq: alpha_p}) in our example.

For a BCM on an ordinary graph, the distance between steady-state opinion clusters is always at least as large as the confidence bound $c$~\cite{ben2003bifurcations, lorenz2007continuous}, which is the theoretical minimum isolation distance in \eqref{eq: triangle} with the dyadic discordance (i.e., the distance between two opinions). For a BCM on a hypergraph, we observe analogous results for all examined discordance functions.

The nucleation and disappearance of opinions clusters has a periodic pattern.
Major opinion clusters (whose population mass is more than $0.1$) and minor opinion clusters (whose population mass is less than or equal to $0.1$) appear in an interlacing fashion. 
The mass of the center opinion cluster grows as $D$ increases until it splits into two symmetric major clusters. 
The center cluster appears again as a minor cluster after the two innermost major clusters are sufficiently far apart; the center cluster grows gradually into a major cluster and repeats the above nucleation and disappearance process.


\subsection{Steady-state densities with different initial opinion distributions}\label{sec: gaussian}
Thus far, we have considered initial opinions that follow a uniform distribution. We now examine the positions of steady-state opinion clusters when opinions initially follow a centered normal distribution $\mathcal{N}(0,\sigma^2)$. That is,
\begin{equation}
    g_0(a) = \frac{1}{\sqrt{2\pi}\sigma} e^{-a^2/\sigma^2}\,.
\end{equation}
We examine the steady-state cluster positions for different values of the distribution variance. 
We use a symmetric computational domain $[-D', D']$, which we take to be sufficiently large so that $g_0(D') < 10^{-9}$.
In Figure \ref{fig: variance}, we show the steady-state cluster positions for different initial variances $\sigma^2$. We fix the confidence bound to $c = 1$ and use the $L_2$ discordance function for all of our simulations. As a comparison, we also include our results with uniform initial distributions.

\begin{figure}[htbp]
  \centering
  \label{fig: variance}
  \includegraphics[width=0.46\textwidth]{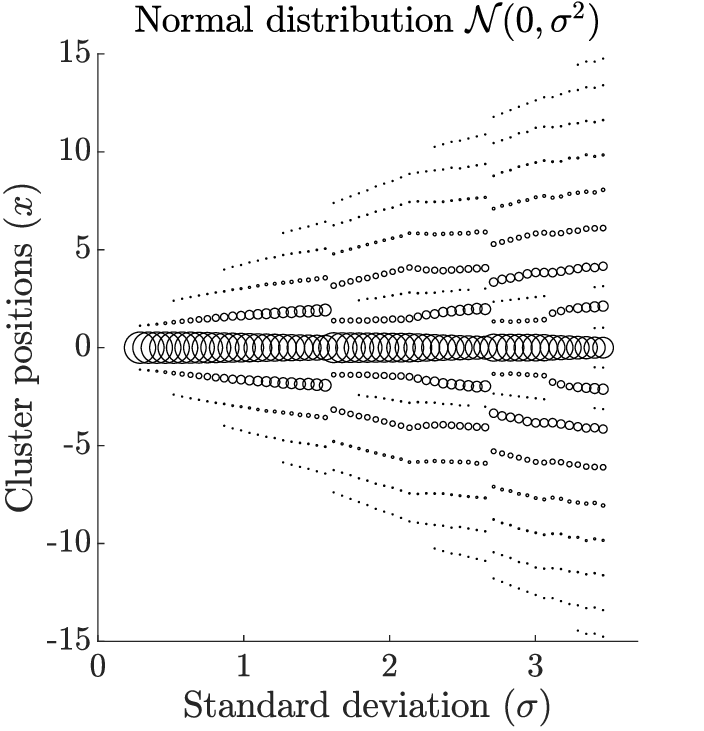}
  \includegraphics[width=0.46\textwidth]{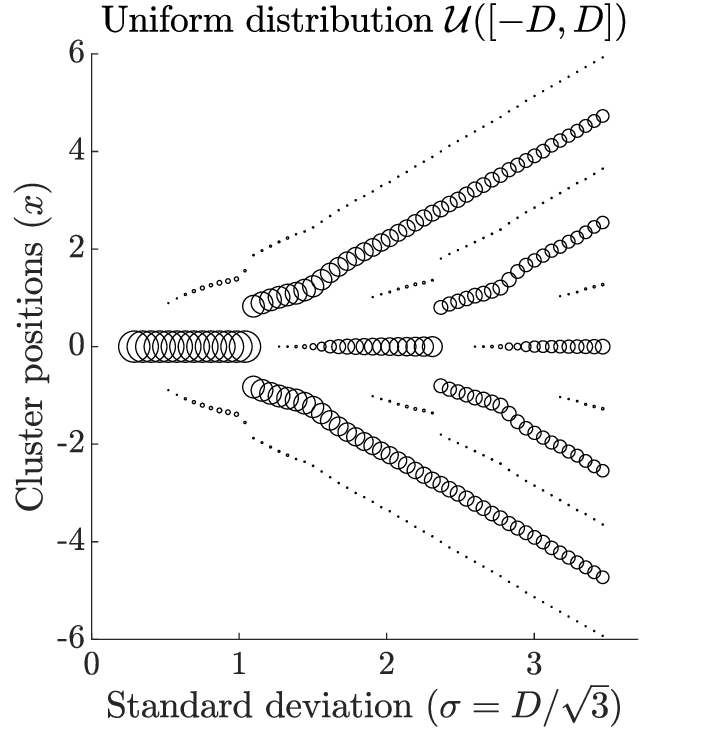}
  \caption{The positions of the steady-state opinion clusters for (left) a normal initial opinion distribution and (right) a uniform initial opinion distribution for different values of the standard deviation $\sigma$.
  The sizes of the circles are proportional to the population masses of the limit opinion clusters. We use $500$ grid points for all simulations and use an adaptive scheme for numerical simulations independently for the separate patches (see Figure \ref{fig: separator}).
  }
\end{figure}

We also observe an interesting pattern in the positions of the steady-state opinion clusters as we increase the initial variance $\sigma^2$. The left panel of Figure \ref{fig: variance} has three ``episodes": $\sigma \in [0.3,1.55]$, $\sigma \in [1.6,2.65]$, and $\sigma \in [2.7,3.5]$. (We consider values of $\sigma$ in increments of $0.05$.) Each episode has a major opinion cluster in the center (i.e., at $a = 0$). In each episode, the center-cluster mass decreases and the minor-cluster masses increase as we increase $\sigma$. Additionally, in each episode, minor clusters nucleate on both sides away from the center. Because the normal distribution is positive on the entire real line, it is possible that the steady state has infinitely many minor clusters and that we simply do not see them with our current computational accuracy. 


\subsection{The steady-state density of agent-based dynamics}\label{sec: meanfieldlimit}

We now examine steady states of the agent-based dynamics with the opinion update rule \eqref{eq: agent-iterations} on a complete network with initial opinions that we draw from a uniform distribution. We sample the steady-state ensemble using $10000$ Monte Carlo simulations. We observe convergence to a sum of Dirac delta functions as the network size increases.
We show the results of these computations in Figure \ref{fig: varyingN}.

\begin{figure}[htbp]
  \begin{center}
 \includegraphics[width=0.8\textwidth]{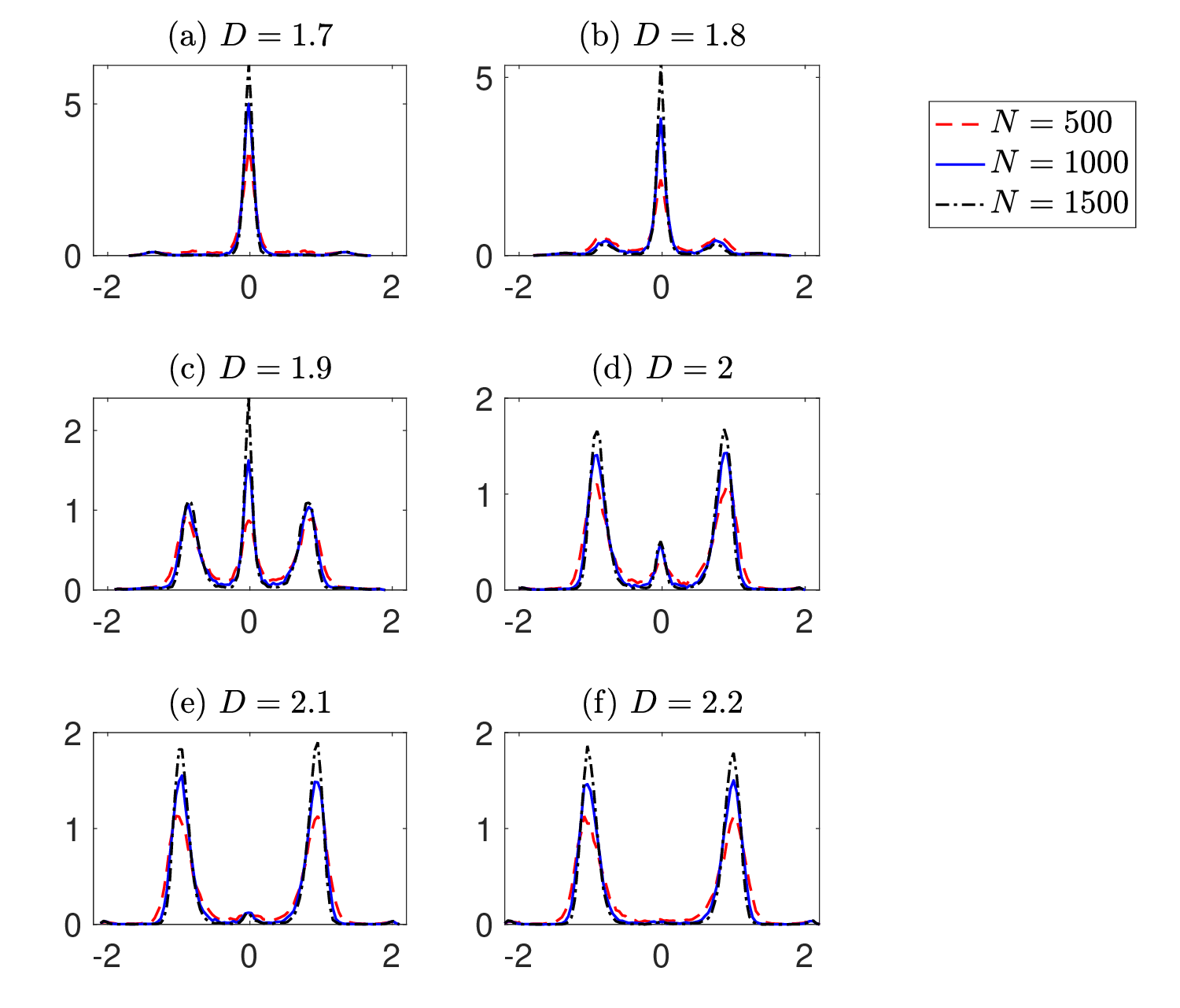}
 \vspace{-10pt}
  \caption{The steady-state \add{PDF} for simulations with finite numbers of agents with the opinion update rule \eqref{eq: agent-iterations} on a complete network.
\add{We obtain the PDF profiles from Monte Carlo simulations of \eqref{eq: agent-iterations} and normalize each PDF so that it integrates to $1$.}
  In all simulations, we draw the initial opinions from the uniform distribution on $[-D,D]$, use the discordance function $d_1$, and take the confidence bound to be $c = 1$. In each panel, we fix $D$ (and hence the initial opinion range) and compare the steady-state densities for $N = 500$, $N = 1000$, and $N = 1500$ agents. In each panel, the horizontal axis is the opinion $x$, and the vertical axis is the steady-state PDF $P(x)$.
  We show results for (a) $D = 1.7$, (b) $D = 1.8$, (c) $D = 1.9$, (d) $D = 2$, (e) $D = 2.1$, and (f) $D = 2.2$.  
  }
  \end{center}
  \label{fig: varyingN}
\end{figure}

For any finite hypergraph size $N$, the density that we obtain from Monte Carlo simulations is not a sum of Dirac delta functions (or a single Dirac delta function) at steady state. As we increase $N$, we observe a slow progression of the densities towards a sum of Dirac delta functions. As we increase $D$ from $1.7$ to $2.2$, the major cluster in the center splits into two symmetric clusters. 
In the density-based simulation in Figure \ref{fig: varyingD}, this split occurs between $1.8$ and $1.9$; this value does not agree quantitatively with what we obtain using Monte Carlo simulations. 

Inconsistencies between density-based models and Monte Carlo simulations of agent-based dynamics have been observed previously in BCMs.
For example, Pineda et al.~\cite{pineda2009noisy} studied a density-based DW model with dyadic interactions and 
used agent-based Monte Carlo simulations to approximate the dynamics on a network 
in the mean-field limit. They observed (see Figure 3 of their paper) that the density evolution can result in qualitatively different behavior (e.g., different numbers of steady-state opinion clusters) as one increases the number of agents in Monte Carlo simulations and that Monte Carlo simulations yield different results than their density-based model.
Because of finite-size effects, the densities from Monte Carlo simulations of their agent-based BCM did not agree with the solutions of their density-based BCM.


\section{Conclusions and discussion}\label{sec: discussion}

We formulated a model of the opinion density of a discrete-time bounded-confidence model (BCM) on hypergraphs and derived the associated rate equation in a mean-field limit with infinitely many agents. By employing a time-rescaling parameter, we extended this rate equation to a continuous-time setting. We thereby obtained a density-based, continuous-time {mean-field} BCM in the form of an integro-differential equation.

We investigated the properties of our density-based BCM by exploring the Cauchy problem of the rate equation in the mean-field limit. We proved that its solution satisfies the properties of a probability density (specifically, nonnegativity and mass conservation). We then examined the asymptotic dynamics of this solution as time goes to infinity. We proved that when the initial condition of the rate equation is a probability density, the solution converges in time in a weak sense to a sum of Dirac delta functions. 
These Dirac delta functions are isolated from each other at least by a constant, which is independent of the initial density and is determined by the discordance function and the confidence bound.

We also studied the steady states of our density-based BCM by numerically solving our integro-differential rate equation for different choices of initial distributions and discordance functions. Our numerical results agree with our mathematical analysis that steady-state solutions take the form of a sum of isolated Dirac delta functions, which yield opinion clusters in the BCM. For both bounded and unbounded initial opinion distributions, we observed that the positions of the steady-state opinion clusters undergo a periodic sequence of bifurcations as we increase the variance of the initial opinion distribution. We also numerically compared our density-based results with agent-based Monte Carlo simulations. As we increased the network size, we observed that the steady-state opinion density converges to a sum of Dirac delta functions.

In our study, we generalized the Deffuant--Weisbuch BCM, but it is also possible to examine other types of BCMs (such as the Hegselmann--Krause model)
on hypergraphs using a density description. In our numerical simulations, we examined $3$-uniform hypergraphs (in which all hyperedges have exactly $3$ nodes), but it is desirable to consider density-based BCMs on more general types of hypergraphs, such as hypergraphs that have both $2$-node and $3$-node hyperedges. We expect that it will be interesting to examine the qualitative dynamics and bifurcation patterns for such hypergraphs.


\section*{Acknowledgements}
MAP acknowledges support from the National Science Foundation (grant number 1922952) through the Algorithms for Threat Detection (ATD) program.


\appendix

\section{Appendix} \label{app}

\subsection{Proof of Theorem \ref{thm: Pdiscrete}}\label{sec: thmPdiscrete}

\begin{proof}
At each time step $n$, we define a discrete random variable $\mathcal{E}(n)$ to describe the selection of a hyperedge $e$ with probability $\mathbb{P}(\mathcal{E}(n)=e) = p_e$.  
We expand $P(\bx,n+1)$ using conditional expectations:
\begin{equation}
\begin{aligned}
    P(\bx, n+1) &= \sum_{e\in E} p_e P(\bx, n+1 | ~ \mathcal{E}(n) = e)\,,
\end{aligned}
\end{equation}
where $P(\bx, n+1 | ~ \mathcal{E}(n) = e)$ is the conditional probability density of selecting hyperedge $e$ at time $n$. 

Using the iterative relation of opinion states in equation \eqref{eq: agent-iterations}, we obtain
\begin{equation}
	\begin{aligned}
       P(\bx, n+1 | ~ \mathcal{E}(n) = e) =&  \int_{d_p(\by_e)<c} P(\by, n) \delta(\bxe-\overline{\by}_e)\delta(\bxte-\byte) \, \mathrm{d}\by \\
       &\!+ \int_{d_p(\by_e)\ge c} P(\by, n) \delta(\bx-\by) \, \mathrm{d}\by\,.
\end{aligned}
\end{equation}
A direct computation yields
\begin{equation}
    \begin{aligned}
           P(\bx, n+1)-P(\bx, n) &= \sum_{e\in E} p_e \Bigg\{ \int_{d_p(\by_e)<c} \!\!\! P(\by, n) \delta(\bxe-\overline{\by}_e)\delta(\bxte-\byte) \, \mathrm{d}\by \\
       &\qquad + \int_{d_p(\by_e)\ge c} \!\!\! P(\by, n) \delta(\bx-\by) \, \mathrm{d}\by - \int P(\by, n) \delta(\bx-\by) \, \mathrm{d}\by \Bigg\}\\
           &= \sum_{e\in E} p_e\int_{d_p(\by_e)<c} \!\!\! P(\bye,\bxte,n)\left[\delta(\bxe-\overline{\by}_e)-\delta(\bxe-\bye)\right] \, \mathrm{d}\bye\,,
    \end{aligned}
\end{equation}
which completes the proof.
\end{proof}


\subsection{Proof of Theorem \ref{thm: gequation}}\label{sec: thmgequation}

\begin{proof} 
A direct computation yields
\begin{equation}
{\small
	\begin{aligned}
    \left[\tau \right]{g}(a,t) &= \frac{\gamma}{N} \sum_{e\in E} p_e\int_{d_p(\by_e)<c} \!\!\!\!\!\!  f(\bye,\bxte,t) \left(\sum_{i=1}^N \delta(a-x_i)\right)\left[\delta(\bxe-\overline{\by}_e)-\delta(\bxe-\bye)\right] \mathrm{d}\bx \, \mathrm{d}\bye \\
    &= \frac{\gamma}{N} \sum_{e\in E} p_e \int_{d_p(\by_e)<c} \!\!\!\!\!\! f(\bye,\bxte,t) \left(\sum_{i\in e} \delta(a-x_i)\right)\left[\delta(\bxe-\overline{\by}_e)-\delta(\bxe-\bye)\right] \mathrm{d}\bx \, \mathrm{d}\bye \\
    &= \frac{\gamma}{N} \sum_{e\in E} p_e \int_{d_p(\by_e)<c}\!\!\!\!\!\!f(\bye,t) \left(\sum_{i\in e} \delta(a-x_i)\right)\left[\delta(\bxe-\overline{\by}_e)-\delta(\bxe-\bye)\right] \mathrm{d}\bxe \, \mathrm{d}\bye \\
    &= \frac{\gamma}{N} \sum_{e\in E} p_e \int_{d_p(\by_e)<c}\!\!\!\!\!\!f(\bye,t) \sum_{i\in e} \left[\delta(a-\overline{\by}_e)-\delta(a-y_i)\right] \mathrm{d}\bye \\
    &= \frac{\gamma}{N} \sum_{e\in E} p_e \int_{d_p(\by_e)<c}\!\!\!\!\!\!f(\bye,t) \left[|e|\delta(a-\overline{\by}_e)-\sum_{i\in e}\delta(a-y_i)\right] \mathrm{d}\bye \\
    &= \frac{\gamma}{N} \sum_{e\in E} p_e \int_{d_p(\by_e)<c}\!\!\!\!\!\!\Pi_{i\in e} g(y_i,t)\left[|e|\delta(a-\overline{\by}_e)-\sum_{i\in e}\delta(a-y_i)\right] \mathrm{d}\bye\,.
\end{aligned}
}
\end{equation}
\quad
\end{proof}


\subsection{A strategy to prove Conjecture \ref{thm: conj}}\label{sec: thmconj}
We believe that it is possible to rigorously prove Conjecture \ref{thm: conj} using approaches that are similar to those in \cite{graham1997stochastic, gomez2012bounded}. 
In this subsection, we connect our work to theirs by presenting a proof strategy for Conjecture \ref{thm: conj} using our paper's language.

Instead of using a one-agent density function $g(a,t)$ to describe a population of agents, References~\cite{graham1997stochastic, gomez2012bounded} characterized the collective behavior of agents using an empirical measure
\begin{equation}
    \Lambda^N = \frac{1}{N}\sum_{i=1}^N\delta_{X_i^N}
\end{equation}
that is induced by a Markov process.
References~\cite{gomez2012bounded,graham1997stochastic} formulated the independence assumption that we use in Theorem \ref{thm: conj} as the problem of \emph{propagation of chaos}, which originates from the study of equivalent interacting particles in statistical physics {\cite{kac1959probability,chaintron2022propagation}}.
For the propagation of chaos, each particle in any finite set of particles yields an i.i.d.~nonlinear Markov process in the mean-field limit \cite{kac1959probability}.

A similar convergence analysis to \eqref{eq: error} was discussed in Theorem 4.3 of \cite{gomez2012bounded} and in Theorem 3.1 of \cite{graham1997stochastic}. One of the main assumptions in those papers is that the size of any set of agents that one is studying (e.g., by examining their joint distribution) must be finite and does not grow as one increases the system size $N$. This assumption corresponds to the constraint in Conjecture \ref{thm: conj} that the maximum hyperedge size $M$ is finite. 


\subsection{Proof of Theorem \ref{thm: convergence}}\label{sec: convergence}

To prove Theorem \ref{thm: convergence}, we start with a lemma.

\begin{lemma}
\label{thm: isolationdistance}
Let $A=\{a_k\}_{k=1}^K$ be a set of discrete values, where $K$ can either be finite or infinite.
Let 
\begin{equation}     \label{eq: XA}
    X_A := \{ \bxe \in A^{|e|} : d_p(\bxe) > 0\} 
\end{equation}
be a set of vectors. It then follows that
\begin{equation} \label{eq: minimizer}
    \min_{\bxe \in X_A} d_p(\bxe) = d_p((A_{\min},0,\ldots,0)) \,,
\end{equation}
where $A_{\min}:=\min_{i\neq j} |a_i - a_j|$.
Up to a permutation, the minimizer of \eqref{eq: minimizer} is of the form
\begin{equation}
    (a_i,a_j,\ldots,a_j) = \mathrm{arg\,min}_{\bxe \in X_A} d_p(\bxe)\,,
\end{equation}
where $|a_i-a_j|=A_{\min}$.
\end{lemma}

\medskip

We now prove Theorem \ref{thm: convergence}.

\begin{proof}
The proof of Theorem \ref{thm: convergence} has three steps: (1) verification of the existence of the limit (i.e., steady-state) distribution; (2) verification that the limit distribution takes the form of a sum of Dirac delta functions; and (3) verification that any two Dirac delta functions must differ by at least the minimum isolation distance $\triangle_{d_p}^*(c)$ in \eqref{eq: triangle}.

\medskip

\underline{\textbf{Step 1.}} We define a sequence of moment generating functions (MGFs) $M_t(s)$ that are associated with the time-dependent probability density $g(a,t)$ by
\begin{equation}
    M_t(s) := \int e^{sa}g(a,t) ~\mathrm{d}a\,.
\end{equation}
Because $e^{sa}$ is convex, we know by Lemma \ref{thm: convex} that $M_t(s)$ is nonincreasing with respect to time. By Theorem \ref{thm: nonnegativity}, we know that $g(a,t) \geq 0$, which implies that $M_t(s) \geq 0$ for all times $t$.
The Monotone Convergence Theorem then guarantees that the infinite-time limit of $M_t(s)$ always exists. We define
\begin{equation}
    M_{\infty}(s) := \lim_{t\rightarrow \infty} M_t(s)\,.
\end{equation}
Because of the unique correspondence between an MGF and a PDF, we know that the infinite-time limit $g_{\infty}(a)$ of $g(a,t)$ also exists. 

\medskip

\underline{\textbf{Step 2.}} As we know from Theorem \ref{thm: conservation}, the mean opinion value $\mu$ is constant with respect to time. 
We define the time-dependent variance of the probability density $g(a,t)$ to be
\begin{equation}
    V(t) := \int (a-\mu)^2g(a,t) \, \mathrm{d}a \,.
\end{equation}
For clarity, we define the shorthand notation $G(\bye,t):=\Pi_{i\in e} g(y_i,t)$. A direct computation yields
\begin{equation}
\begin{aligned}
    \frac{\partial}{\partial t}{V}(t) &= \sum_{e\in E} p_e \int_{d_p(\bye)<c} \!\!\!\!\!\!G(\bye,t) (a-\mu)^2  \left[ |e|\delta(a-\overline{\by}_e) - \sum_{i\in e}\delta(a-y_i)\right] \mathrm{d}a \, \mathrm{d}\bye\\
    &= \sum_{e\in E} p_e \int_{d_p(\bye)<c}\!\!\!\!\!\!G(\bye,t)\left[ |e|(\overline{\by}_e-\mu)^2 - \sum_{i\in e}(y_i-\mu)^2\right] \mathrm{d}\bye\\
    &= - \sum_{e\in E} \frac{p_e}{2|e|} \int_{d_p(\bye)<c}\!\!\!\!\!\!G(\bye,t)\sum_{i,j\in e} (y_i-y_j)^2 \, \mathrm{d}\bye\,.
\end{aligned}    
\label{eq: dotV}
\end{equation}
Because $V(t)$ converges to the variance of $g_{\infty}(a)$, we know that $\lim_{t\rightarrow\infty}\frac{\partial}{\partial t}V(t)=0$, which yields
\begin{equation}    \label{eq: integralg}
    \int_{d_p(\bye)<c} \Pi_{i\in e} g_{\infty}(y_i)\sum_{i,j\in e} (y_i-y_j)^2 ~\mathrm{d}\bye = 0
\end{equation}
for all hyperedges $e$.

Let $U$ denote the ``positive set" of $g_{\infty}$ (i.e., $U = \{a:g_{\infty}(a)>0\}$). If the Lebesgue measure of $U$ is positive, then we can find an interval $\Omega_c=(y^*, y^*+c)$ such that the overlap region $O$ of $U$ and $\Omega_{c}$ also has a positive Lebesgue measure. 
Furthermore, for any $\bye \in O^{|e|}\backslash \{\bye: y_i=y_j  \,\, \text{for all} 
i,j \in e\}$, we have
\begin{equation}
    \mathbb{1}_{d_p(\bye)<c} \Pi_{i\in e} g_{\infty}(y_i)\sum_{i,j\in e} (y_i-y_j)^2 > 0 \,,
\end{equation}
which implies that the integral in \eqref{eq: integralg} over $O^{|e|}\backslash \{\bye: y_i = y_j \, \, {\text{for all}} \,\,
i,j \in e\}$ is strictly positive. 
This contradicts equation \eqref{eq: integralg} and proves that $g_{\infty}$ has an ``empty" positive set. In other words, $U$ has a $0$ measure.
Additionally, $g_{\infty} \geq 0$ and integrates to $1$, so $g_{\infty}$ is a Dirac measure and we can write $g_{\infty}(a) = \sum_{k=1}^K m_k\delta(a - a_k)$, with $\sum_{k=1}^K m_k = 1$ because of conservation of mass (see Theorem \ref{thm: conservation}).

\medskip

\underline{\textbf{Step 3.}} 
Let $A$ be the discrete set $\{a_k\}_{k=1}^K$. We insert $g_{\infty}$ into \eqref{eq: integralg} and obtain
\begin{equation}
    \sum_{\by_a \in X_A} \mathbb{1}_{d_p(\by_a)<c}\Pi_{i=1}^{|e|}m_{k_i} \sum_{i,j=1}^{|e|}(a_{k_i}-a_{k_j})^2 = 0\,,
\end{equation}
where $\by_a=(a_{k_1},\ldots,a_{k_{|e|}})$ and $X_A$ are defined in \eqref{eq: XA}. Because $m_{k_i}\sum_{i,j=1}^{|e|}(a_{k_i} - a_{k_j})^2 > 0$ holds for all $\by_a \in X_A$, it follows then 
$\mathbb{1}_{d_p(\by_a) < c} = 0$ for all $\by_a \in X_A$. This is equivalent to the statement that $\min_{\by_a\in X_A}d_p(\by_a) \ge c$. 
By Lemma \ref{thm: isolationdistance}, we know that $\min_{\by_a\in X_A}d_p(\by_a) = d_p((A_{\min},0,\ldots,0))$, where $A_{\min}=\min_{i\neq j}|a_i-a_j|$. By the definition of the minimum isolation distance  $\triangle_{d_p}^*(c)$ in \eqref{eq: triangle}, we know that $d_p((\triangle_{d_p}^*(c),0,\ldots,0))=c$, which implies that
\begin{equation}
    d_p((A_{\min},0,\ldots,0)) \ge d_p((\triangle_{d_p}^*(c),0,\ldots,0)) \,.
\end{equation}
Consequently, $A_{\min} \ge \triangle_{d_p}^*(c)$.

When $g_{0}(a)$ has a finite support $B$, it is also true that $a_i \in B$. Because $|a_i - a_j| > \triangle_{d_p}^*(c)$, it follows that $K$ is finite and $K \le \floor{|B|/\triangle_{d_p}^*(c)}$. This completes the proof.

\quad
\end{proof}


\bibliographystyle{siamplain}
\bibliography{references}


\end{document}